\documentclass[final]{siamltex}

\usepackage{latexsym, amsmath, amssymb}

\newcommand{\comment}[1]{}

\newtheorem{fact}[theorem]{Fact}
\newtheorem{remark}[theorem]{Remark}

\def\rk{{\rm rk\;}}

\def\F{{\mathbb F}}

\def\Aalg{{\cal A}}
\def\Balg{{\cal B}}

\def\Hom{{\rm Hom}}

\def\Lin{{\rm Lin}}
\def\Elin{{\rm Lin}}
\def\Env{{\rm Env}}
\def\Ann{{\rm Ann}}

\def\rem#1{}
\def\sset{{\cal S}}

\title{Deterministic 
Polynomial Time Algorithms for Matrix Completion Problems}
\author{
G\'abor Ivanyos
\thanks{Computer and Automation Research Institute
of the Hungarian Academy of Sciences (MTA SZTAKI),
L\'agym\'anyosi u. 11,
1111 Budapest. 
E-mail: {\tt Gabor.Ivanyos@sztaki.hu}}
\and
Marek Karpinski
\thanks{Department of Computer Science and Hausdorff Center for Mathematics, 
University of Bonn, 53117 Bonn.
E-mail: {\tt marek@cs.uni-bonn.de}}
\and 
Nitin Saxena
\thanks{Hausdorff Center for Mathematics,
Endenicher Allee 62, University of Bonn, 
53115 Bonn.
E-mail: {\tt ns@hcm.uni-bonn.de}}
}%author

\begin{document}
\maketitle

\begin{abstract}
We present new deterministic algorithms for several cases of the {\em maximum rank matrix completion}
problem (for short {\em matrix completion}), i.e. the problem of assigning values to the variables in
a given symbolic matrix as to maximize the resulting matrix {\em rank}. Matrix completion belongs to
the fundamental problems in computational complexity with numerous important algorithmic applications,
among others, in computing dynamic transitive closures or multicast network codings \cite{hkm05,hky06}.
 We design efficient deterministic algorithms for common generalizations of the results
of Lov\'{a}sz and Geelen on this problem by allowing {\em linear polynomials} in the entries of the input
matrix such that the submatrices corresponding to each variable have rank one.

Our methods are algebraic and quite different from those of Lov\'{a}sz and Geelen.
We look at the problem of matrix completion in the more general setting
of linear spaces of linear transformations, and finding a max-rank element there using a 
greedy method. 
Matrix algebras and Modules play a crucial role in the algorithm. 
We show (hardness) results for special instances of matrix
completion naturally related to matrix algebras, namely: in contrast to
computing isomorphism of modules (for which there is a known
deterministic polynomial time algorithm), finding 
a surjective or an injective homomorphism between 
two given modules is as hard as the general 
matrix completion problem. 
%As a result we connect the
%classical problem of polynomial identity testing with checking 
%surjectivity (or injectivity)
%between two given modules.
The same hardness holds for finding a maximum dimension
{\em cyclic} submodule (i.e. generated by a single element).
For the ``dual" task, i.e. finding minimal
number of generators of a given module, we present a deterministic 
polynomial time algorithm. The proof methods developed in this paper apply to fairly 
general modules and could be also of independent interest.
\end{abstract}

%\renewcommand{\thefootnote}{\fnsymbol{footnote}}
%\footnotetext[4]{
%The authors thank the Hausdorff Research Institute for Mathematics, Bonn for the kind support.}



\pagestyle{myheadings}
\thispagestyle{plain}
\markboth{Ivanyos, Karpinski \& Saxena}{Matrix Completion}

\section{Introduction}

A {\em linear matrix} is a matrix having linear polynomials as its entries, say the linear polynomials
are over a field $\F$ and in $\F[x_1,\ldots,x_n]$. The problem of 
{\em maximum rank matrix completion}, or just {\em matrix completion} for
short, is
the problem of assigning values from the field $\F$ to the variables $x_1,\ldots,x_n$ such that the rank 
of a given linear 
matrix is maximized (over all possible assignments). The notion of linear matrices appears in 
several places including both theory and applications, see \cite{hkm05,hky06} for several references.
The problem of matrix completion is a well studied problem, dating back to the work of
Edmonds \cite{e67} and Lov\'{a}sz \cite{l79}. A similar problem
(basically an equivalent one) is
{\em nonsingular matrix completion}, where we have a square linear
matrix and we are interested in an assignment resulting in a nonsingular
matrix. If the ground field is sufficiently large then
the maximum rank achieved by completion coincides with
the rank of the linear matrix considered as a matrix over
the function field $\F(x_1,\ldots,x_n)$, and hence,
by standard linear algebra, finding 
a maximum rank completion (equivalently determining the maximum rank)
is in deterministic polynomial time reducible to instances
of finding (equivalently deciding the existence of) nonsingular completion
of certain minors. Lov\'{a}sz gave an efficient randomized
algorithm to find a matrix completion using the Schwartz-Zippel lemma \cite{s80,z79}, 
deducing that a random assignment of the variables will maximize the rank if the field is large 
enough (see also \cite{im83}). This is a method also useful in the fundamental problem of 
polynomial identity testing 
(PIT). Indeed matrix completion is equivalent to a special case of PIT: any arithmetic formula can 
be written as the determinant of a linear matrix \cite{v79}, hence the formula would be nonzero
iff the corresponding matrix could attain full rank (assuming a large enough field). Over large
fields, this makes matrix completion an important problem in ZPP, as its derandomization
would imply circuit lower bounds (see Kabanets \& Impagliazzo \cite{ki03}). 

Over small fields, matrix completion soon becomes a hard problem. This version has some
important practical applications, for example in constructing multicast network codes \cite{hkm05}, 
and hence there are several results in the literature specifying the exact parameters for 
which the problem becomes NP-hard. The hardness of matrix completion and various related 
problems were first studied by Buss et al. \cite{bfs99} and more recently by Harvey et al.
\cite{hky06}. In the former paper nonsingular matrix completion
is proved to be NP-hard over fields of constant size, while 
the latter showed that matrix completion over the field $\F_2$ is NP-hard 
even if we restrict to a matrix where each variable occurs at most twice in its 
entries. This naturally raises the question: can we solve matrix completion by restricting
the way the variables appear in the input matrix?

Few such cases are already known and they all look at {\em mixed matrices}, i.e. linear 
matrices where each entry is either a variable or a constant. Harvey et al. \cite{hkm05},
building on the works of Geelen \cite{g99} and Murota \cite{m00}, gave an efficient
deterministic
algorithm for matrix completion over {\em any} field if the mixed matrix has each variable
appearing at most once. While Geelen at al. \cite{gim03,gi05} gave an efficient
deterministic
algorithm when the mixed matrix is {\em skew-symmetric} and has each variable appearing 
at most twice.

\smallskip
{\bf Completion by rank one matrices:}
In this paper we are interested in cases that are more general than the first case 
\cite{hkm05}. Consider a linear matrix $A\in\F[x_1,\ldots,x_n]^{m\times m}$ where the 
submatrix ``induced'' by each variable is of rank one, i.e. $A=B_0+x_1B_1+\cdots+x_nB_n$ 
where $B_1,\ldots,B_n$ are constant matrices of rank one (note that $B_0$ is also a 
constant matrix but of arbitrary rank). The case $B_0=0$ was first considered 
by Lov\'{a}sz in \cite{l89}, where it is shown how Edmonds' matroid intersection 
algorithm can be applied to solve this special case in deterministic 
polynomial time. The first main result in this paper is 
a common generalization of the results of Lov\'asz \cite{l89}
and Geelen \cite{g99}: we show that matrix completion 
problem for an {\em arbitrary} $B_0$ can be solved in deterministic polynomial 
time over any field:

\begin{theorem}\label{thm-main-1} 
Let $\F$ be a field and let $B_0,\ldots,B_n$ be $m\times m$
matrices over $\F$. If $B_1,\ldots,B_n$ are of rank one then matrix completion for 
the matrix $(B_0+x_1B_1+\cdots+x_nB_n)$ can be done deterministically in $poly(m,n)$
field operations.
\end{theorem}

The proof of this theorem basically involves looking at the linear space 
$L:=\langle B_0,B_1,\ldots,B_n\rangle$ of matrices and showing that a greedy approach can be 
utilized to gradually increase the rank of an element in $L$. Our methods are more algebraic 
and quite different 
from those of Lov\'{a}sz and Geelen. In particular our method is 
robust enough to check whether a given matrix in $L$ has the largest possible rank 
{\em without} needing the rank one generators of $L$, 
they are needed only if we want to increase the rank 
(see Section \ref{sec-augment}). 

Matrix algebras or algebras of linear transformations
(in this paper by an {\em algebra} we mean 
a linear space of matrices or linear transformations 
that is also closed under multiplication)
play a crucial role in the algorithm for Theorem~\ref{thm-main-1}. We consider 
special instances of matrix completion problems where algebras of linear
transformations arise naturally. These are certain {\em module} problems.

\smallskip
{\bf Preliminaries about modules: }
If $U$ and $V$ are vector spaces over the field $\F$ then
we denote the vector space of linear maps from $U$ to $V$
by $\Lin(U,V)$. For $\Lin(U,U)$ we use the notation
$\Elin(U)$.
For simplicity, in this paper we consider modules over finite sets.
(Actually, we work with modules over free associative algebras, however
the main concepts and computational tasks we are concerned with
can be understood without any knowledge from the theory 
of abstract associative algebras.) 
Let $\sset$ be a finite set. A vector space $V$ over 
the field $\F$ equipped with a map $\nu$ from $\sset$ into
$\Elin(V)$ is
called an {\em $\F\{\sset\}$-module} (or an $\sset$-module for short if
$\F$ is clear from the context). We assume that the data for an 
$\sset$-module is input by an $|\sset|$-tuple of $\dim V$ by
$\dim V$ matrices. In cases when the map 
$\nu$ is clear from the context - most typically when
$\sset$ is itself a set of linear transformations -
we omit $\nu$ and denote the result $\nu(B)v$ of the action of 
$B\in \sset$ on $v\in V$ by $Bv$. For a set 
$\sset'\subseteq \Elin(V)$ of linear transformations 
the {\em enveloping algebra} $\Env(\sset')$
is the smallest algebra containing $\sset'$. It is the linear
span of finite products of transformations from $\sset'$ (and maybe 
noncommutative).  

In the context of $\sset$-modules the algebra 
$\Aalg=\Env(\nu(\sset)\cup I)$ is of special interest ($I$ is the 
identity in $\Elin(V)$).
An $\sset$-{\em submodule} of $V$ is a linear subspace
closed under the action of all the transformations in $\nu(\sset)$.
Obviously, the intersection of a family of submodules is
again a submodule. In particular, if $T$ is a subset 
of $V$ then there is a smallest submodule of $V$ containing
$T$: the submodule generated by $T$. It is $\Aalg T$,
the linear span of vectors obtained by application of
transformations from $\Aalg$ to vectors from $T$.
The set $T\subseteq V$ is a system of {\em generators} for the 
$\sset$-module $V$ if $V=\Aalg T$. 

{\em Cyclic} submodules, i.e. those  generated by a single element,
are of particular interest. For $v\in V$ we consider
the map $\mu_v:\Aalg\rightarrow V$ given by
$\mu_v(B)=Bv$. Obviously, $\mu_v$ is a linear
map from $\Aalg$ into $V$ and the set $\{\mu_v|v\in V\}$
is a linear space of linear maps from $\Aalg$ to $V$.
The {\em rank} of $\mu_v$ is the dimension of the submodule
$\Aalg v$ generated by $v$. 

\smallskip
{\bf A ``Universal" Module Problem: }
The matrix completion problem
in this context is finding an element $v$ which generates
a submodule of maximum dimension. It turns out that
this problem, which we call {\em cyclic submodule optimization}, 
is universal in matrix completion: 
there is a deterministic polynomial time reduction
from maximum rank matrix completion to cyclic submodule optimization 
(over an arbitrary base field).
We show this universality in Section~\ref{sec-morph}. 
Universality implies two hardness results. First,
existence of a deterministic polynomial time algorithm for cyclic submodule 
optimization would imply deterministic solvability of
the matrix completion problem over sufficiently large fields.
Also, over small fields, cyclic submodule optimization
is $NP$-hard. Second, we get analogous hardness results
for the existence of injective resp.~surjective homomorphisms
between modules (a $\sset$-module {\em homomorphism} from $V$ to $V'$ is a 
linear map in $\Lin(V,V')$ that commutes with the action of $\sset$):

\begin{theorem}\label{thm-main-3}
There is a deterministic polynomial time reduction from the existence 
of (resp. finding) a nonsingular 
matrix completion to the problem of checking for the existence 
of (resp. finding) a surjective (or injective) homomorphism 
between two modules.
\end{theorem}

This result is remarkable
in view of the recent deterministic polynomial time
algorithm of Brookbanks \& Luks \cite{bl08}
for module isomorphism problem (see also Chistov et al. \cite {cik97} over
special base fields).

\smallskip
{\bf A ``Dual" Problem: }
In Section \ref{sec-min-gen} we consider a
problem which is in some sense ``dual" to the cyclic submodule 
optimization. This is finding  
a system of generators of smallest size for a module.
In contrast to hardness of the former problem,
 we have an efficient solution to the latter: 

\begin{theorem}\label{thm-main-2}
Given a module structure
on the $n$-dimensional vector space $V$ over the field $\F$ 
in terms of $m$ $n\times n$ matrices $\sset$, one can find the minimum number of 
generators of $V$ deterministically using $poly(m,n)$ field operations.
\end{theorem}

Note that the above result includes testing cyclicity of modules efficiently
over any field. 
This problem was considered in \cite{cik97} over special fields
as a tool for constructing isomorphisms between modules.
The algorithm is based on a greedy approach analogous to the method for
Theorem~\ref{thm-main-1}, implicitly using certain submodule dimension
optimization technique for a special class of (so called {\em semisimple})
modules.

\section{Matrix Completion with Rank One Matrices}\label{sec-augment}

Let $V$ be a finite dimensional vector space
over the field $\F$ and $L\leq \Elin(V)$
be a $\F$-linear space of linear transformations.
Recall that $\Env(L)$, the {\em enveloping algebra} of $L$ is
the linear span of products 
$h_1h_2\cdots h_s$ ($s\geq 1,\; h_1,\ldots,h_s\in L$). 
Obviously, $\Env(L)$ is also spanned by products 
of elements from an arbitrary basis of $L$. 
We will use the 
action of the enveloping algebra on the kernel of an idempotent transformation  
to optimize rank in a linear space, to that effect we present the following lemma.
Its proof will also suggest how to greedily increment the rank.

\begin{lemma}\label{lem-augment}
Let $V$ be a finite dimensional vector space
over the field $\F$, let $L\leq \Elin(V)$ and assume that $e\in L$ is an 
idempotent ($e^2=e$) such that $\rk e\geq \rk h$ for every $h\in L$.
If $L$ is spanned by $e$ and certain rank one 
transformations, then $\Env(L)\ker e\subseteq eV$.
\end{lemma}
\begin{proof}
Assume, for contradiction, that $\Env(L)\ker e$ is not contained in
$eV$. Then there exists a vector $v\in\ker e$ such that $L^sv\not\subseteq eV$
for some integer $s$. Let $s\geq 1$ be the smallest 
among such integers. Then there are matrices
$h_1,\ldots,h_s\in L$ with $h_s\cdots h_2h_1v\not\in eV$ 
such that for every $i$, 
the matrix $h_i$ is either $e$ or has rank
one. Assume that
$h_j=e$ for some $j\leq s$. Then $j>1$ as $ev=0$.
Furthermore, the minimality of $s$ implies 
\begin{equation}\label{eqn-in-eV}
h_{j-1}\cdots h_1v\in eV, 
\end{equation}
therefore, as $ew=w$ for every $w\in eV$, we have
$h_s\cdots h_{j+1} h_jh_{j-1}\cdots h_1v= h_s\cdots$ $h_{j+1}eh_{j-1}\cdots h_1v=$ 
$h_s\cdots h_{j+1}h_{j-1}\cdots h_1v$,
contradicting the minimality of $s$. Thus all the matrices
$h_1,\ldots,h_s$ are of rank one.
Set $v_0=v$ and for $1\leq i\leq s$,
$v_i=h_iv_{i-1}$. The minimality of
$s$ implies that for every $1\leq i\leq s$ we have
$v_i\in L^iv\setminus \sum_{j=0}^{i-1}L^jv$.
In particular, the vectors $v_0,\ldots,v_s$
are linearly independent. Since $h_i$ is a rank one transformation on $V$, 
\begin{equation}\label{eqn-rank-1-action}
h_iV=\F v_i, \text{ for all } 1\le i\le s. 
\end{equation}
From this, and from the minimality 
of $s$ we infer $h_jv_{i-1}=0$ for
every $1\leq i<j\leq s$ (otherwise $v_j\in h_j L^{i-1}
v\subseteq L^i v$). We show below that $a:=e+h_1+\ldots+h_s$ is of
a rank higher than $e$, leading to the desired contradiction.

Informally, we build a basis in which the matrix of $a$ is upper triangular. 
First we see that (keep in mind Equations \ref{eqn-in-eV} \& \ref{eqn-rank-1-action}) 
$av_{i-1}=
ev_{i-1}+\sum_{j<i}h_jv_{i-1}+v_i+\sum_{j>i}h_jv_{i-1}=
v_i+ev_{i-1}+\sum_{j<i}h_jv_{i-1}\in v_i+e v_{i-1}+\sum_{j<i}\F v_j 
\subseteq v_i+\langle v_1,\ldots,v_{i-1}\rangle$ ($i=1,\ldots,s$).
(With some abuse of notation, for $i=1$, 
by $\langle v_1,\ldots,v_{i-1}\rangle$ we mean the zero subspace.) 
Hence the vectors $av_0,\ldots,av_{s-1}$ span
the subspace $\langle v_1,\ldots,v_s \rangle$.
Let $W$ be a direct complement to
the subspace $\langle v_1,\ldots,v_{s-1}\rangle$
in $eV$ and $w_1,\ldots,w_t$
be a basis of $W$ ($t=\rk e-s+1$). Then
$aw_i=ew_i+\sum_{j=1}^sh_jw_i=
w_i+\sum_{j=1}^sh_jw_i\in w_i+\langle v_1,\ldots,v_s\rangle$ (by Equation 
\ref{eqn-rank-1-action}).
Therefore the vectors $aw_1,\ldots,aw_t$ are linearly
independent even modulo the subspace $\langle v_1,\ldots,v_s\rangle$.
Together with the fact that $\langle av_0,\ldots,av_{s-1}\rangle=
\langle v_1,\ldots,v_s \rangle$, this implies that
$\langle av_0,\ldots,av_{s-1},aw_1,\ldots,aw_t\rangle
=\langle v_1,\ldots,v_s,w_1,\ldots,w_t\rangle=\langle eV,v_s\rangle.$
Thus the image of $a$ contains a subspace of dimension $\rk e+1$
and hence $\rk a\geq \rk e+1$, as claimed.
\end{proof}

In the proof above, the special case $s=1$ deserves special attention.
In that case we have a simple method for increasing the rank 
over sufficiently large fields which works even {\em without} any 
assumption on the presence of rank one matrices. We will use
this simple observation later in Section~\ref{sec-min-gen}.

\begin{lemma}\label{lem-imker}
If $h,h''\in \Lin(U,V)$ are transformations such that
$h''\ker h\not\subseteq hU$ then
$h':=h+\alpha h''$ will be of a higher rank than $h$
except for at most $\rk h+1$ elements $\alpha\in \F$.
\end{lemma}
\begin{proof}
Let $k=\rk h$, and let $U_0$ be a subspace of $U$ complementary
to $\ker h$. Let $u_1,\ldots,u_k$ be a basis of $U_0$ and
let $v_1,\ldots,v_k$ be a basis of the image $hU$. Choose
a vector $u_{k+1}\in\ker h$ such that $v_{k+1}:=h'' u_{k+1}\not
\in hU$. Consider the matrix of the restriction of $h+xh''$ 
to $U_0+\F u_{k+1}$ in the bases $u_1,\ldots,u_k,u_{k+1}$ and
$v_1,\ldots,v_k,v_{k+1}$. The last row of the constant
term (the matrix of $h$) is zero while the lower right entry of the linear
term (the matrix of $xh''$) is $x$. Expanding by the last
row, we obtain that the linear term of the determinant of this $(k+1)$
by $(k+1)$ matrix is $dx$, where $d\neq 0$ is the determinant of
the upper left $k\times k$ block of $h$. 
Thus the determinant
is a nonzero polynomial in $x$ of degree at most $(k+1)$
and hence the corresponding
$(k+1)$ by $(k+1)$ block of $h'=h+\alpha h''$ is nonsingular
showing that $h'$ has rank higher than $k$ unless
$\alpha$ is a root of this polynomial. 
\end{proof}

We state below a simple fact about the linear spaces of matrices
that is useful in providing a certificate for the rank maximality of a given 
matrix.

\begin{fact}\label{fac-maxrk-ineq}
Let $L\leq \Lin(U,V)$, where $U$ and $V$ are
finite dimensional spaces over the field $\F$.
Then for every $h\in L$ we have
$\rk h\le\dim U-max\{\dim W-\dim LW\ |\ W\leq U\}$.
\end{fact}
\begin{proof}
For any subspace $W\leq U$ 
pick a direct complement $W'$ of $W$ in $U$.
Now $\dim U-\rk h$ $=\dim U-\dim hU$ $\ge (\dim W-\dim hW)+
(\dim W'-\dim hW')$ $\ge \dim W-\dim LW$.
\end{proof}

Using Edmonds' Matroid Intersection Theorem,
Lov\'asz (Section 3, \cite{l89}) has shown that equality holds
provided that $h$ is of maximum rank and if $L$ is spanned by rank one matrices.
We give the following algorithmic generalization to the case when $L$ is 
spanned by rank one matrices and an {\em arbitrary} rank matrix.

\begin{theorem}\label{thm-witness}
Let $U$ and $V$ be two finite dimensional vector spaces
over the field $\F$, let $L\leq \Lin(U,V)$ be given
by a basis and let an $h\in L$ be also given. Suppose that 
$L$ is spanned by $h$ and certain (unknown) transformations of rank one.

1)
Then there exists a deterministic polynomial time algorithm which
decides if $h$ is an element of $L$ of maximum rank. 
If $h$ is of maximum rank then a subspace $W$ of $U$
is constructed such that 
$\rk h=\dim U-(\dim W-\dim LW)$.

2)
If $h$ is not of maximum rank then, given
rank one transformations that together with $h$ span $L$, we can compute an element $h'\in L$
with $\rk h'>\rk h$ in deterministic polynomial time. 
\end{theorem}

\begin{proof}
We may assume wlog that $\dim U=\dim V$, for otherwise we can 
pad transformations from $L$ with zeros to obtain a space
$L'\leq \Lin(U\oplus U',V\oplus V')$, where 
$\dim U\oplus U'=\dim V\oplus V'$ with some (possibly zero)
spaces $U',V'$. By padding a transformation $b\in \Lin(U,V)$
we mean the map $b'\in \Lin(U\oplus U',V\oplus V')$ which is the direct
sum of $b$ and the zero map: $b'(u,u')=(bu,0)$. 

Let $g:V\rightarrow U$ be an arbitrary nonsingular linear map such that 
$gh:U\rightarrow U$ is an idempotent. (The matrix of such a map $g$
can be obtained as the product of the matrices
corresponding to the pivoting steps in Gaussian elimination for
the matrix of $h$.)
%(It is easy to construct such a $g$: 
%let $V''=\ker h$ and let $V'$ be a direct complement of $V''$ in $V$
%(a subspace such that $V=V'+ V''$ and $V'\cap V''=(0)$). Then $hV=hV'$. 
%Let $g_1:hV'\rightarrow V'$ be the inverse of the restriction
%of $h$ to $V'$, let $U''$ be a direct complement of $hV$ and 
%let $g_2:U''\rightarrow V''$ be any invertible linear map. Then
%the direct sum $g=g_1\oplus g_2$ defined as $g(u_1+u_2)=g_1u_1+g_2u_2$
%($u_1\in hV$, $u_2\in U''$) will do.).
As $g$ is invertible, $h$ is of maximum rank within $L$ iff $gh$
is of maximum rank within $gL$. Also, rank one 
generators of $L$ are mapped to rank one
generators of $gL$. If $gh$ is of maximum rank then by 
Lemma~\ref{lem-augment}, $\Env(gL)\ker gh\leq ghU$.
Conversely if $\Env(gL)\ker gh\leq ghU$
then, with $W_0:=\Env(gL)\ker gh$ and $W_1:=\ker gh$, we have 
$gLW_0, gL W_1\leq W_0\le ghU$, and $W_0\cap W_1$=0 (if $v\in W_0\cap W_1$
then $v=ghu$ for some $u\in U$ and $ghv=0$, implying $0=ghghu=ghu=v$).
Therefore with $W:=W_0+W_1\le U$ we have $gL W\leq W_0$
and $\dim U-\rk gh=\ker gh=\dim W_1=\dim W-\dim W_0\le \dim W-\dim gLW$.
Now $g$ being invertible also implies that $\dim U-\rk h\le \dim W-\dim LW$,
which together with Fact \ref{fac-maxrk-ineq} implies that $h$ has maximal 
rank. Thus if $\Env(gL)\ker gh\leq ghU$ then we can efficiently 
construct $W$ with the required property, it is a witness of the maximality
of the rank of $gh$ (resp. $h$) in $gL$ (resp. $L$). 
Thus, $h$ and hence $gh$ is not of maximum rank 
if and only if $\Env(gL)\ker gh$ is not contained in $ghU$. 
This can be decided in an obvious way. 

Furthermore, if $L$ is spanned by $h$ and (known) rank one
matrices $h_1,\ldots,h_\ell$ then 
the proof of Lemma~\ref{lem-augment}
gives a linear combination
of $gh$ and $gh_1,\ldots,gh_\ell$ of higher rank.
Multiplying by $g^{-1}$ we obtain 
an element of $L$ of rank larger than $\rk h$.
\end{proof}

It is obvious that repeated applications of Theorem \ref{thm-witness} 
completes the proof of Theorem~\ref{thm-main-1}.

We remark that the shortest product $\Pi=gh_1\cdots gh_\ell$ 
with $\Pi \ker h\not\subseteq hU$ in the proof of Lemma~\ref{lem-augment}
can be interpreted as a generalization of 
the notion of {\em augmenting paths} in the classical bipartite
matching algorithms.

\section{Module Morphism Problems and Matrix Completion}\label{sec-morph}

In this section we present hardness results
of certain problems concerning modules. The key
constructions are modules that we call 
{\em bipartite modules} as they resemble bipartite graphs.

\subsection{Bipartite modules}

Let $W_1$ and $W_2$ be two linear spaces over $\F$ and assume that
we are given a linear subspace $R\leq \Lin(W_1,W_2)$ of
linear maps from $W_1$ to $W_2$. We assume that $R$ is spanned
by $\ell$ maps: $r_1,\ldots,r_\ell$. We consider the direct sum
$W=W_1\oplus W_2$. We extend transformations $r\in R$ to linear
transformations of $W$ by letting $r$ act on $W_2$ as the zero map. 
(That is, the extension maps $(w_1,w_2)$ to $(0,rw_1)$.)
With some abuse of notation, we denote the extended map also by 
$r$ and consider $R$ as a subspace of $\Elin(W)$. 
Let $\sset$ be the set $\{r_1,\ldots,r_\ell\}$
and the map $\nu:\sset\rightarrow \Elin(W)$ defining
the $\sset$-module structure be just the identity map. This $\sset$-module
$W$ is a bipartite module.

\subsection{Universality of cyclic submodule optimization}

Assume that we are given a linear space $L$ of $\F$-linear maps
from $U$ to $V$. It would be straightforward to consider
the bipartite module $W$ for $W_1=U$, $W_2=V$ and $R=L$.
However, this module does not turn out to be useful for our purposes
and instead of it we consider another view: put $W_1=L$, $W_2=V$, 
$R=\{\mu_u|u\in U\}$, where $\mu_u(h)=h\cdot u$. Then, if
$U$ is spanned by $u_1,\ldots,u_\ell$ then $\sset$
is $\{\mu_{u_1},\ldots,\mu_{u_\ell}\}$. 
Let $(h,v)\in W=L\oplus V$. Then the $\sset$-submodule
of $W$ generated by $(h,v)$ is $\F(h,v)+(0,hU)$
and its dimension is $(1+\rk h)$ if $h$ is not the zero map.
Therefore this construction transforms matrix completion
in $L$ to cyclic submodule optimization in $W$.

\subsection{Module morphisms}

Let $U$ and $V$ be two $\F\{\sset\}$-modules. 
An $\F$-linear map $\phi\in \Lin(U,V)$ 
is an $\sset$-module {\it homomorphism} if
for every $s\in\sset$ and $u\in U$ we have $\phi(su)=s\phi(u)$.
The module homomorphism from $U$ to $V$ form a linear
subspace $\Hom_{\F\{\sset\}}(U,V)$ of $\Lin(U,V)$. 
Given the $\sset$-module structure on $U$ and $V$ in
terms of matrices over bases, a basis for the matrix space
representing $\Hom_{\F\{\sset\}}(U,V)$ can be computed with
$poly(\dim U+\dim V+|\sset|)$ field operations by
solving a system of homogeneous linear equations.

It is not difficult to construct subspaces of
$\Lin(U,V)$ which do not arise as spaces
of module homomorphisms. Thus it is natural
to ask how difficult are the matrix completion problems
in spaces of module morphisms. It turns out (as shown below) that
the cyclic submodules of a bipartite module $W$ (defined as $L\oplus V$ in the last subsection)
arise as homomorphic images of another $\sset$-module $W_0$, 
where $\sset=\{r_1,\ldots,r_\ell\}$ and
$W_0$ has basis $b_0,b_1,\ldots,b_\ell$ that by definition satisfy
$r_ib_0=b_i$, $r_ib_j=0$ ($i,j=1,\ldots,\ell$).

This shows that hard matrix completion problems
do arise in module morphism spaces. However, curiously
enough, deciding existence and construction of module 
{\it isomorphisms}, i.e., module homomorphisms 
which are bijective linear maps can be accomplished
in polynomial time (see \cite{cik97} with some restriction
for the base field and \cite{bl08} over arbitrary fields). 
We show that this is not the case for testing
existence of injective or surjective module morphisms.
 
\smallskip
{\bf Module Injection: } 
For the injective case, consider the bipartite modules
$W$ and $W_0$ discussed above. The module $W_0$ is cyclic,
it is generated by $b_0$. Therefore a module homomorphism
is determined by the image of $b_0$. In this case for
every pair $(w_1,w_2)$ there is indeed a homomorphism
with $\psi(b_0)=(w_1,w_2)$. (For $i>0$ set $\psi(b_i)=(0,r_iw_1)$.)
Consider the special case of the bipartite module
$W$ used for showing hardness of cyclic submodule  
optimization: let $L$ be a space of linear maps from $U$ to $V$,
put $W_1=L$ and $W_2=V$. Then the image of $W_0$ at the map
$\psi$, under which the image of $b_0$ is $(h,v)$, is
the subspace spanned by $(h,v),(0,hu_1),\ldots,(0,hu_\ell)$.
This $\psi$ is injective if and only if $h$ is. This
construction reduces both deciding and finding an injective transformation
in $L$ (and also nonsingular matrix completion as special case)
to deciding and finding an injective homomorphism from $W_0$ to $W$. 

\smallskip
{\bf Module Surjection: }
Existence of (resp. finding) injective
module morphisms can be transformed to the existence of
(resp. finding) surjective morphisms between 
modules by standard {\em dualization}. If $M$ is
a vector space over $\F$ then by $M^*$ we denote
the space of (homogeneous) linear functions from 
$M$ to $\F$ (that is, $M^*=\Lin(M,\F)$). If $\phi$ is an $\F$-linear map from
the space $M_1$ to $M_2$ then the map $\phi^*:M_2^*\rightarrow M_1^*$ 
given as $(\phi^*f)v=\phi (f v)$ is again a linear map.
(Note that if $\phi$ is interpreted as multiplication
of column vectors by a matrix from the left then 
$\phi^*$ can be interpreted as multiplication of row
vectors by the transposed matrix from the right.)
Furthermore, if both $M_1$ and $M_2$ are finite dimensional then
$\phi$ is injective (resp. surjective) if and only if
$\phi^*$ is surjective (resp. injective). If $M_1$ and $M_2$
are $\sset$-modules given by the maps $\nu_1$ and $\nu_2$,
then $\nu_1^*$ and $\nu_2^*$ given as $\nu_i^*(s)=\nu_i(s)^*$
make $M_1^*$ and $M_2^*$ $\sset$-modules.
Furthermore, the linear map $\phi\in \Lin(M_1,M_2)$ is
a module homomorphism from $M_1$ to $M_2$ if and only if
$\phi^*$ is a module homomorphism from $M_2^*$ to $M_1^*$.

So when given vector spaces $U, V$ over $\F$, a linear subspace $L$ of 
$\Lin(U,V)$ with $\ell:=\dim U$. We first construct modules
$W$ and $W_0$ as in the previous reduction, so that the
module homomorphism $\psi$ from $W_0$ to $W$ is injective
if and only if $h\in L$ is an injective map where $\psi(b_0)=(h,v)$.
Therefore $\Psi\in \Hom_{\F\{\sset\}}(W^*,W_0^*)$ is surjective 
if and only if for the unique $\F$-linear map $\psi:W_0\rightarrow W$
such that $\Psi=\psi^*$ we have that $h$ is injective,
where $\psi(b_0)=(h,v)$. This completes the proof of Theorem~\ref{thm-main-3}.

\section{Minimizing Number of Generators in Modules}\label{sec-min-gen}

We saw that cyclic submodule optimization is matrix completion hard. Now we 
will study the ``dual" problem of finding minimal number of generators of
a given module. 
In this section we give an {\em efficient} algorithm to minimize the number of generators 
in a given $\F\{\sset\}$-module. It depends on a greedy property of the dimension 
of submodules in so called semisimple modules 
(which will be vaguely similar to that in 
Section~\ref{sec-augment}). But we first need to summarize some basic notions and facts from 
the representation theory of algebras needed in the proof. For details, we refer
the reader to the first few chapters of the textbook \cite{p82}.

\comment{The {\em radical} of a module $V$ 
is the intersection of its maximal 
(more accurately, maximal proper) submodules. A
finite dimensional module is {\em semisimple} if its radical is the zero submodule. 
Submodules, direct sums and factor modules of semisimple
modules as well as the factor of a module by its radical 
are semisimple. (Recall that submodules,
direct sums of spaces which are $\sset$-modules inherit the $\sset$-module 
structure in a natural way and so do factor spaces 
by submodules.)}%end comment

\subsection{Preliminaries: Algebras, Modules \& their Decompositions}

Let $\F$ be an arbitrary field.
An associative algebra with identity or {\em algebra} for short
 is a vector space $\Aalg$ over $\F$ equipped with an associative 
$\F$-bilinear multiplication having a two-sided identity element 
$1_\Aalg$ with respect to the multiplicative structure. 
If $V$ is a finite dimensional vector space of $\F$ then the $\F$-linear 
transformations of $V$ form a finite dimensional algebra $\Elin(V)$.
Subalgebras of $\Elin(V)$, that is, subspaces closed under 
multiplication, containing the identity matrix are further
examples. (In contrast to Section~\ref{sec-augment}, where we considered
algebras of linear transformations not necessarily having an identity,
in this Appendix it will be convenient to consider algebras with
identity only.) An algebra {\em homomorphism} from $\Aalg$ to $\Balg$
is an $\F$-linear map $\phi:\Aalg\rightarrow\Balg$ 
also satisfying $\phi(a_1\cdot a_2)=\phi(a_1)\cdot \phi(a_2)$
and $\phi(1_\Aalg)=1_\Balg$.

A left $\Aalg$-module or an {\em $\Aalg$-module} for short is an 
$\F$-linear space $V$ equipped with a bilinear multiplication 
$\cdot:\Aalg\times V\rightarrow V$ which commutes with the
multiplication within $\Aalg$ 
(that is, $a_1\cdot(a_2\cdot v)=(a_1\cdot a_2)\cdot v)$).
(In \cite{p82}, right modules are used. Here we 
we use left modules which are somewhat more common
in the literature.) A module $V$ is {\em unital} if $1_\Aalg v=v$ for
every $v\in V$. All modules in this work are assumed 
to be unital and {\em finite dimensional} over $\F$. 

If $V$ is an $\Aalg$-module then the map 
$\nu:\Aalg\rightarrow \Elin(V)$ defined as $\nu(a)v=a\cdot v$ is
a homomorphism from $\Aalg$ into $\Elin(V)$.
We say that $V$ is a {\em faithful} $\Aalg$-module
if the kernel of $\nu$ is zero, that is,
if $a\in\Aalg$ such that $av=0$ for every $v\in V$ then $a=0$. 
If $\sset$ is a finite set then $\F\{\sset\}$,
the algebra of noncommutative polynomials over $\F$
with indeterminates from $\sset$ is an example 
of an infinite dimensional $\F$-algebra. It is the {\em free} algebra
generated by $\sset$: if $\Aalg$ is an algebra and $\nu$ is a map
from $\sset$ into $\Aalg$ then $\nu$ can be extended to
a unique algebra homomorphism from $\F\{\sset\}$ to $\Aalg$.
In view of this, an $\F\{\sset\}$-module structure on $V$
can be given by an arbitrary map $\nu:\sset\rightarrow \Elin(V)$.
Thus the notion of $\sset$-module used in this paper is
consistent with the notion of modules over free algebras.

A {\em submodule} of an $\Aalg$-module is a linear subspace
also closed under multiplication by elements of $\Aalg$.
The {\em factor space} of a submodule inherits the $\Aalg$-module
structure in a natural way and so do {\em direct sums} of linear
spaces which are $\Aalg$-modules. An $\Aalg$-module $V$ is 
called {\em simple} if it has exactly two submodules: the whole
$V$ and the zero submodule. The {\em radical} of a module is
the intersection of its maximal (more precisely, maximal proper)
submodules. A module $V$ is called {\em semisimple} if it is 
isomorphic to a direct sum of simple modules. By Section~2.7 of \cite{p82}, 
$V$ is semisimple if and only if its radical is the zero submodule.
Furthermore, the factor module of $V$ by its radical is always 
semisimple. By Section~2.5 of \cite{p82}, the isomorphism
classes of the constituents and their multiplicities in a decomposition
of a semisimple module into a direct sum of simple modules
are uniquely determined. Direct sums and homomorphic images of semisimple
modules are semisimple.

Let $V$ be a finite dimensional $\F\{\sset\}$-module
and let $\Aalg$ be the enveloping algebra $\Env(I\cup\nu(\sset))$
(the subalgebra of $\Elin(V)$ generated by the identity and $\nu(\sset)$). 
Then $\Aalg$ is the image
of $\F\{\sset\}$ under the unique algebra homomorphism from $\F\{\sset\}$
to $\Elin(V)$ extending $\nu$ and $V$ is a faithful
$\Aalg$-module in the natural way. We work with the $\Aalg$-module structures of $V$, 
its submodules and factors as they
coincide with the $\sset$-module structures of the same objects.
Assume that $V$ is semisimple.
Then by Section~4.1 of~\cite{p82}, $\Aalg$ considered as a left module
over itself by the algebra multiplication is also semisimple. 
Such algebras are called {\em semisimple}. Modules over semisimple 
algebras are semisimple, again by Section~4.1 of~\cite{p82}. 

Let $\Aalg$ be a semisimple algebra over $\F$ and let
$\Aalg$ as a left module over itself be isomorphic to
the direct sum: 
\begin{equation}\label{eqn-A}
\bigoplus_{i=1}^{t}V_i^{m_i}, 
\end{equation}
where $V_i$ are pairwise non-isomorphic 
$\Aalg$-modules. Let $V$ be an $\Aalg$-module. 
As $V$ is a homomorphic
image of at most $\dim V$ copies of
the module $\Aalg$, we have
\begin{equation}\label{eqn-V}
V\cong \bigoplus_{i=1}^{t}V_i^{s_i},
\end{equation}
where the multiplicities $s_i$ are non-negative integers.

\begin{lemma}
\label{lem-max-dim}
Let $\Aalg$ and $V$ be as above and let $\ell$ be a positive
integer. Let $U$ be a submodule of $V$ generated by $\ell$ elements.  
Then $U$ is of maximum dimension among the $\ell$-generated 
submodules of $V$ if and only if
$U\cong \bigoplus_{i=1}^{t}V_i^{d_i}$ where $d_i:=\min (s_i,\ell m_i)$.
\end{lemma}

\begin{proof}
Let $W_i$ be the sum of all simple submodules of $V$ 
not isomorphic to $V_i$, Then the $\Aalg$-module 
$V/W_i$ is isomorphic to $V_i^{s_i}$ and that the submodule dimension in $V$ 
is maximized iff it is maximized in $V/W_i$ for all $i\in[t]$. As a single 
generator in $V_i^{s_i}$ can generate a submodule of dimension at most 
that of $V_i^{\min(s_i,m_i)}$, we get that $\ell$ generators 
in $V_i^{s_i}$ can generate a submodule of dimension at most 
that of $V_i^{d_i}$. Repeating this for
every $i\in\{1,\ldots,t\}$, we obtain that the maximum dimension is at most 
the dimension of the direct sum in the statement. 

To see that this module occurs in fact as a cyclic submodule of $V$,
let $W$ be the direct sum of $\ell$ copies of $\Aalg$ 
(as a left $\Aalg$-module) and let $w_1=(1_\Aalg,0,\ldots,0),$ $\ldots,$
$w_\ell=(0,\ldots,0,1_\Aalg)$. Let $W_0$ be a submodule of $W$ isomorphic
to $\bigoplus_{i=1}^t V_i^{\ell m_i-d_i}$ and let $V_0$ be a submodule
of $V$ isomorphic to $\bigoplus_{i=1}^{t}V_i^{d_i}$. Then
$V_0\cong W/W_0$ and $W/W_0$ is generated by $\ell$ elements:
the images of $w_1,\ldots,w_\ell$ under the projection $W\rightarrow W/W_0$.
Thus $V_0$ can be generated by the images of the latter $\ell$ elements
under any isomorphism $W/W_0\cong V_0$.
\end{proof}

\subsection{A Greedy Optimization of the Submodule Dimension in Semisimple Modules}

In this section $V$ denotes a finite dimensional
$\F\{\sset\}$-module and $\Aalg$ stands for the 
enveloping algebra $\Env(\nu(\sset)\cup I)$. 
For subsets $\Balg\subseteq \Aalg$ and $U\subseteq V$
by $\Balg U$ we denote the linear span
of the products $bu$, where $b\in \Balg$ and $u\in U$.
In this context we omit braces around one-element sets.
In particular, for $v\in V$, the submodule generated by
$v$ is $\Aalg v$.

The {\em annihilator} $\Ann_\Aalg(U)$ of 
$U\subseteq V$ is $\{a\in \Aalg|au=0\mbox{~for every~}u\in U\}$.
Note that the annihilator $\Ann_\Aalg(v)$ of the single element 
$v\in V$ is  just the kernel of the linear map $\mu_v:\Aalg\rightarrow V$
given as $\mu_v(a)=av$. The following lemma states that
if the rank of $\mu_v$ is not maximal then we are in the situation of
Lemma~\ref{lem-imker}. 

\begin{lemma}\label{lem-cycrank}
Assume that $V$ is semisimple. Then, for an arbitrary $u\in V$, 
$\dim \Aalg u=\max\{\dim \Aalg u'|u'\in V\}$ iff
$\Ann_\Aalg(u)V\subseteq \Aalg u$. 

Furthermore, if 
$\Ann_\Aalg(u)V\not\subseteq \Aalg u$ then an element $u'$ with
$\dim \Aalg u'>\dim \Aalg u$ can be constructed 
 using $poly(|\sset|+\dim V)$ operations in $\F$.
\end{lemma}
\begin{remark}
The lemma generalizes a result of Babai and R\'onyai
which was used in \cite{br90} for solving the
cyclic submodule optimization in modules over 
simple algebras.   
The proof can be found in \cite{cik97}. For completeness,
we discuss it here as well. 
The second part of the 
lemma is especially interesting for small base fields where
Lemma~\ref{lem-imker} does not apply.
\end{remark}
\begin{proof}
Let $V$ be a semisimple $\sset$-module and let
$\Aalg=\Env(I\cup \nu(\sset))$. Let $\Aalg$
resp.~$V$ be decomposed as in (\ref{eqn-A}) resp.~(\ref{eqn-V}).
Let $u\in V$. Assume that the dimension of the submodule
$\Aalg u$ is not maximal. Then, by Lemma~\ref{lem-max-dim},
there exists an index $i$ such that the multiplicity of
$V_i$ in $\Aalg u$ is less than both $s_i$ and $m_i$.
Let $W$ be the submodule of $V$ which is the
direct sum of the constituents of $V$ not isomorphic 
to $V_i$. Then $V/W\cong V_i^{s_i}$ and
$V/(W+\Aalg u)$ is isomorphic $V_i^h$ with
some $h>0$. Recall that for a subset $X$ of $V$
the annihilator of $X$ in $\Aalg$, denoted by
$\Ann_\Aalg(X)$ is $\{a\in\Aalg|ax=0\mbox{~for every~}x\in X\}$.
Assume that $\Ann_\Aalg(u)V\subseteq \Aalg u$.
Then every element of $\Ann_\Aalg(u)$ act as zero
on the factor module $V/\Aalg u$ and hence also on
the factor $V/(W+\Aalg u)$. As the latter module is isomorphic
to $V_i^h$ we obtain that $\Ann_\Aalg(u)\subseteq \Ann_\Aalg(V_i)$.
Recall that the map $\mu_u:\Aalg\rightarrow V$ is given
as $\mu_u(a)=au$. It is an $\Aalg$-module homomorphism
from the left module $\Aalg$ to $V$. Its kernel is
$\Ann_\Aalg(u)$ and its image is $\Aalg u$. Therefore
$\Aalg u\cong \Aalg/\Ann_\Aalg(u)$. Now $\Ann_\Aalg(V_i)$
is also an $\Aalg$-submodule of $\Aalg$. Let $L$ be a submodule
of $\Aalg$ isomorphic to $V_i$. We claim that $LV_i\neq 0$.
Indeed, if $LV_i=0$ then, by the assumed isomorphism, $LL=0$
as well, which is impossible by Section 3.2 of~\cite{p82}.
The claim implies that the multiplicity of $V_i$
in $\Ann_\Aalg(V_i)$ is zero and the same holds
in $\Ann_\Aalg(u) \subseteq \Ann_\Aalg(V_i)$. But then
the multiplicity of $V_i$ in the factor module
$\Aalg/\Ann_\Aalg(u)\cong \Aalg u$ is $m_i$. 
This contradiction finishes the proof of:
if $\Aalg u$ is not of maximum dimension then
in fact $\Ann_\Aalg(u)V\not\subseteq \Aalg u$.

To see the reverse implication,
assume that $\Ann_\Aalg(u)V\not\subseteq \Aalg u$
and let $w\in V$ and $b\in \Ann_\Aalg(u)$ such that
$bw \not\in \Aalg u$. 
 By Section 2.4 of \cite{p82}, there exists 
a submodule $W'$ of $V$ such that $W'\cap \Aalg u=0$
and $W'+\Aalg u=V$. Write $w=au+w'$ where $a\in\Aalg$
and $w'\in W'$. Put $u'=u+w'$. As $\Aalg w'\in W'$,
we have $\Aalg u'+W'= \Aalg u+W'$. On the other hand,
from $bw \not\in \Aalg u$ but $bau\in\Aalg u$ we infer
that $bw'$ is a nonzero element of $W'$ and by 
the equality $bu'=bu+bw'=bw'$, it is also an element of $\Aalg u'$.
Therefore $\dim \Aalg u'>\dim V-\dim W'=\dim \Aalg u$,
as required. 

For a polynomial time implementation of the
construction above, notice that a basis for
$\Ann_\Aalg(u)$ can be found by solving 
a system of linear equations. Then $b$ and $w$
can be found by testing membership of 
products of pairs
of basis elements for $\Ann_\Aalg(u)$ and those for
$V$. To compute a direct complement of
$\Aalg u$, we first compute a projection
$\pi$ of $V$ onto $\Aalg u$ such that
$\pi a= a \pi$ for every element $a\in\Aalg$
(equivalently, for every element of
a system of generators for $\Aalg$,
say $\nu(\sset)$). (Recall that a projection
$\pi$ onto a subspace $V'$ of $V$ is
a map whose image is $V'$ and it acts
as the identity on $V'$. If $W'$
is submodule complementary to $\Aalg u$
then the unique linear map which is the identity on
$\Aalg u$ and zero on $W'$ is a projection
onto $\Aalg u$ which commutes with
the action of $\Aalg$ on $V$.)
Once $\pi$ is constructed we take
$\pi'=I-\pi$. It is straightforward to see
that the image $W'=\pi'V$ is in fact
a direct complement of $\Aalg u$. The
element $w'$ in the argument above is
then just $\pi'w$ and $u'=u+\pi'w$.
This finishes the proof of Lemma~\ref{lem-cycrank}.
\end{proof}

The next lemma can be used to give a generalization
for submodules generated by larger systems (eg. noncyclic modules).

\begin{lemma}\label{lem-multrank}
Assume that $V$ is semisimple. Then, for arbitrary positive integer $\ell$
and for elements $u_1,\ldots,u_\ell\in V$, 
$\dim \Aalg\{u_1,\ldots,u_\ell\}=
\max\{\dim\Aalg U|U\subseteq V,\#{U}\leq \ell\}$
if and only if for every $i\in [\ell]$, the $\sset$-submodule
generated by $u_i+W_i$ in the factor module $V/W_i$
is of maximum dimension, where $W_i$ denotes the submodule
generated by $u_1,\ldots,u_{i-1},u_{i+1},\ldots,u_\ell$.
\end{lemma}
\begin{proof}
Let $V$ be a semisimple $\sset$-module and let
$\Aalg=\Env(I\cup \nu(\sset))$. Let $\Aalg$
resp.~$V$ be decomposed as in (\ref{eqn-A}) resp.~(\ref{eqn-V}).
Let $u_1,\ldots,u_\ell\in V$ and let 
$W_i=\Aalg(\{u_1,\ldots,u_\ell\}\setminus \{u_i\})$.
As $\Aalg\{u_1,\ldots,u_\ell\}=\Aalg u_i+W_i$,
it is obvious that if, for some index $i$ there is an element $u_i'$ such
that modulo $W_i$, $\Aalg u_i'$ has a larger dimension than $\Aalg u_i$,
then replacing $u_i$ with $u_i'$ results in a system generating
a submodule of larger dimension. 

To see the reverse implication let $W=\Aalg\{u_1,\ldots,u_\ell\}$
and assume that for every $i$, the submodule of $V/W_i$
generated by $u_i+W_i$, that is, $W/W_i$ is a maximal
dimensional cyclic submodule of $V/W_i$. Let
$j\in \{1,\ldots,t\}$. By Lemma~\ref{lem-max-dim},
for every $i$, the multiplicity of $V_j$ in $W/W_i$ 
equals either the multiplicity of $V_j$ in $V/W_i$
or it is just $m_j$. If for some index $i$ the former
is the case then the multiplicity of $V_j$ in $V/W$
is zero. Otherwise the multiplicity of $V_j$ in
$W/W_i$ is $m_j$ for every index $i$. In the former
case the multiplicity of $V_j$ in $W$ is the maximum possible
among all submodules. Assume the latter case and let
$U_{ij}$ denote the direct sum of the constituents
of $\Aalg u_i$ isomorphic to $V_j$. Then, for every index $i$,
we have $U_{ij}\cap W_i =0$ and $U_{ij}$
is isomorphic to a direct sum of $m_j$ copies of $V_j$
as otherwise the multiplicity of $V_j$ in $W/W_i$ would
be less than $m_j$. Thus $U_{ij}$ intersects
$\sum_{i'\neq i}U_{i'j}$ trivially therefore
they form an independent system and hence
$\sum_{i=1}^\ell U_{ij}\cong V_j^{\ell m_j}$ showing
that the multiplicity of $V_j$ is optimal in this
case as well. Repeating this for every irreducible
module $V_j$, we obtain that the dimension
of $W$ is indeed the maximum possible.
This finishes the proof of Lemma~\ref{lem-multrank}.
\end{proof}

The two lemmas above together with Lemma~\ref{lem-imker}
immediately give the following. 

\begin{proposition}\label{pro-cik-gen}
Let $v_1,\ldots,v_n$ be a basis of the semisimple
$\F\{\sset\}$-module $V$. 
Assume that $u_1,\ldots,u_\ell$ are elements of $V$
such that the submodule generated by
$u_1,\ldots,u_\ell$ is not of maximum dimension
among the submodules of $V$ generated by at most $\ell$
elements. If the $\F\{\sset\}$-module structure on
$V$ is given by an array of matrices, then we can find an index $i$
and construct $u_i'\in V$ using $poly(|\sset|+n)$ operations 
such that replacing $u_i$ with $u_i'$
results in a submodule of larger dimension.

Furthermore, if $|\F|>n$ then there exist indices $i\in[\ell]$,
$j\in[n]$ such that replacing $u_i$ with $(u_i+\omega v_j)$
results in a submodule of larger dimension
except for at most $n$ elements $\omega$ from $\F$.
\end{proposition}

The above greedy property for the submodules of a semisimple 
module gives us the following technical lemma for {\em general} modules.
It will be useful in the subsequent algorithm for optimizing the number of 
generators in any module without computing the radical explicitly.

\begin{lemma}\label{lem-cik-compl}
Let $v_1,\ldots,v_n$ be a basis of the
$\sset$-module $V$ 
which can be generated by $\ell$ elements 
and let
$u_1,\ldots,u_\ell$ be elements of $V$
such that $U=\Aalg\{u_1,\ldots,u_\ell\}<V$.
If $W$ is a nonzero submodule such that $V=U\oplus W$
then there exist $i\in[\ell]$, $j\in[n]$ 
such that for 
$U':=\Aalg\{u_1,\ldots,u_i+\lambda v_j,\ldots,u_\ell\}$,
$V=U'+W$ but $U'\cap W\neq\{0\}$
except for at most $2n$ elements $\lambda\in \F$.
\end{lemma}
\begin{proof}
Let $U_0, W_0$ be the radicals of $U, W$ respectively.
Let $V_0=U_0\oplus W_0$. Then the factor module 
$V/V_0\cong U/U_0\oplus W/W_0$ is semisimple
and we can apply the preceding 
Proposition \ref{pro-cik-gen} to choose $i\in[\ell],j\in[n]$ such 
that the number of $\lambda$'s, for which the dimension of 
$(U'+V_0)/V_0$ is not 
larger than the dimension of $(U+V_0)/V_0$, is 
at most $\dim V/V_0$. Also for the same
$i,j$ the $\lambda$'s, for which 
$\{u_1,\ldots,u_i+\lambda v_j,\ldots,u_\ell\}\cup W$ do not span
the whole of $V$, are the roots of a nonzero $\F$-polynomial of degree 
at most $\dim V$. Thus for this $i,j$ the number $\lambda$'s, for which either
$\dim U'\le \dim U$ or $V\ne U'+W$, is at most
$\dim V/V_0+\dim V\leq 2n$. 
\end{proof}

\subsection{Algorithm for Finding $\ell$ Generators}

Using the previous Lemma, now we describe an iterative algorithm to find 
a minimal set of generators of a given module over a sufficiently large
ground field.

\smallskip\noindent
{\bf Input:} An $\Aalg$-module $V$ given in terms of a set of generators. 
We assume that $\Aalg$ is an $\F$-algebra where $|\F|>2\dim V$.
\\
{\bf Output:} A set of at most $\ell$ elements generating $V$ over $\Aalg$.
\\
{\bf Algorithm:}

\begin{itemize}
\item[0)] 
Initially pick any irredundant generating set
$\{u_1,\ldots,u_\ell,u_{\ell+1},\ldots\}$, 
set $U:=\Aalg \{u_1,\ldots,u_\ell\}$ and $W:=\Aalg\{u_{\ell+1},\ldots\}$. Then $V=U+W$. 
\\
{\tt Outer loop:}

\item[1)]
Set $W':=U \cap W$. 

\item[2)]
If $W'=W$ then output $U$ and exit.
\\
{\tt Inner loop:}

\item[3)]
apply Lemma \ref{lem-cik-compl} in $V/W'$ to obtain
$U'$ generated by $\ell$ elements and satisfying: 
\\
\mbox{~~~~~~~~}$U'+W=V$ and $( U'+W')\cap W>W'$.

\item[4)]
If such a $U'$ cannot be found then report ``$\ell$ generators are 
insufficient for $V$'' and exit.
\\
Else set $U:=U'$.

\item[5)]
If $W\not \leq  U+W'$ then continue {\em inner} loop with $W'=(U+W')\cap W$.\\
Else continue {\em outer} loop with $W=W'$.
\end{itemize}

\smallskip\noindent
{\em Analysis of the algorithm:} At each step of the algorithm there is a 
pair $(U,W)$ of $\sset$-modules such that $V=U+W$ and $U$ is known in terms of 
$\ell$ generators. 
At every repetition of the inner loop: $W'$ becomes a
larger submodule of $W$, since at Step 5 we know (from Step 3) that $(U+W')\cap W$
is strictly larger than $W'$. 
At every repetition of the outer loop: $W$ becomes a smaller submodule 
of $V$, since at Step 5 we know (again from Step 3) that $W'$ is strictly smaller
than $W$. Thus, the number of times the algorithm can loop is bounded by $(\dim V)^2$,
which makes the algorithm polynomial time. This
gives a proof of Theorem~\ref{thm-main-2} over large base fields. 

Over small
base fields we use the algorithm of \cite{fr85} or \cite{ciw97} to compute
the radical of $\Aalg$ and the radical $V_0$ of $V$ therefrom and compute
a minimal generating set $\Gamma_0$ of the factor module $V/V_0$ using 
Proposition~\ref{pro-cik-gen} directly. For each $u_0\in \Gamma_0$
we pick a representative $u\in u_0+V_0$ and obtain a subset $\Gamma\subseteq V$
such that $|\Gamma|=|\Gamma_0|$ and $\Gamma\cup V_0$ generates $V$. 
By a standard property of the radical, we show that $\Gamma$ itself
generates $V$. Indeed, let $U$ be the submodule generated by $\Gamma$. 
If $U\neq V$ then there is a {\em maximal} (proper) submodule 
$U'\supseteq U\supseteq \Gamma$. But $U'\geq V_0$ by the definition
of $V_0$, therefore $U'\supseteq \Gamma\cup V_0$, implying $U'\supseteq V$, which
is a contradiction to $U'$ being proper. This ends the proof of Theorem~\ref{thm-main-2}.

\section{Concluding remarks}

We have shown that the maximum rank matrix
in a linear space generated by rank one matrices
and a further matrix of arbitrary rank can be found
in deterministic polynomial time if the rank one generators
are given. It would be interesting to know if there is
an efficient deterministic method in the case where
the rank one generators are not known. In this direction
we have a deterministic polynomial time
 algorithm, which, given a matrix of maximum rank constructs 
a certificate that the rank is in fact maximal (see
Theorem~\ref{thm-witness}) without knowing the rank one generators.
This implies that over sufficiently large base fields,
the maximum rank matrix can be constructed
in {\em Las Vegas} polynomial time. The best result of this flavor
is the deterministic polynomial time algorithm of Gurvits 
\cite{gu03,gu04} which decides whether there exists 
a nonsingular matrix in the space generated by rational
matrices under the assumption that the span over the
complex numbers can be generated by 
unknown rank one matrices (with not necessarily rational
entries). Unfortunately, this algorithm decides the mere existence
of a nonsingular matrix without explicitly constructing one.

The space of the maps $\mu_v:\Aalg\rightarrow V$ where
$V$ is a {\em semisimple} $\sset$-module and $\Aalg$ is the corresponding
enveloping algebra has a curious property that if $\mu_v$ is not
of maximum rank there is a $v''\in V$ such that Lemma~\ref{lem-imker}
applies for $h=\mu_v$ and $h''=\mu_{v''}$ (see Lemma~\ref{lem-cycrank}). 
In particular, over a sufficiently large field $\F$ the rank of $\mu_v+\alpha \mu_{v''}$
will be higher for some $v''$ chosen from an arbitrary
basis of $V$ and a ``generic" $\alpha\in \F$.

It would be interesting to find more classes ${\cal L}$
of spaces of linear maps with such a ``local rank incrementing"
property: There is a constant $c$ such that for every $L\in{\cal L}$,
if $h\in L$ is not of maximum rank then from an arbitrary
basis $h_1,\ldots,h_\ell$ of $L$ one can choose maps 
$h_{i_1},\ldots,h_{i_c},$ such
that $h+\alpha_1 h_{i_1}+\ldots+ \alpha_ch_{i_c}$ 
has higher rank for 
some $\alpha_1,\ldots,\alpha_c\in \F$ ($\F$ is large enough.)

\section{Acknowledgements}

We would like to thank the anonymous referees for several suggestions.
We are grateful to the Hausdorff Research Institute for Mathematics, Bonn for its hospitality
and the kind support.

%\newpage

\bibliographystyle{alpha}
\bibliography{refs}

\begin{thebibliography}{HKM05}

\bibitem[BFS99]{bfs99}
J.~F. Buss, G.~S. Frandsen, and J.~Shallit.
\newblock The computational complexity of some problems of linear algebra.
\newblock {\em J. Comput. Syst. Sci.}, 58(3):572--596, 1999.

\bibitem[BL08]{bl08}
P.A. Brooksbank and E.M. Luks.
\newblock Testing isomorphism of modules.
\newblock {\em Journal of Algebra}, 320(11):4020--4029, 2008.

\bibitem[BR90]{br90}
L.~Babai and L.~R{\'o}nyai.
\newblock Computing irreducible representations of finite groups.
\newblock {\em Mathematics of Computation}, 55(192):705--722, 1990.

\bibitem[CIK97]{cik97}
A.~Chistov, G.~Ivanyos, and M.~Karpinski.
\newblock Polynomial time algorithms for modules over finite dimensional
  algebras.
\newblock In {\em ISSAC '97: Proceedings of the 1997 International Symposium on
  Symbolic and Algebraic Computation}, pages 68--74, 1997.

\bibitem[CIW97]{ciw97}
A.~M. Cohen, G.~Ivanyos, and D.~B. Wales.
\newblock Finding the radical of an algebra of linear transformations.
\newblock {\em Journal of Pure and Applied Algebra}, 117-118:177--193, 1997.
\newblock (Proc. MEGA'96).

\bibitem[Edm67]{e67}
J.~Edmonds.
\newblock Systems of distinct representatives and linear algebra.
\newblock {\em Journal of Research of the National Bureau of Standards},
  71B:241--245, 1967.

\bibitem[FR85]{fr85}
K.~Friedl and L.~R\'{o}nyai.
\newblock Polynomial time solutions of some problems of computational algebra.
\newblock In {\em STOC '85: Proceedings of the seventeenth annual ACM Symposium
  on Theory of Computing}, pages 153--162, 1985.

\bibitem[Gee99]{g99}
J.~F. Geelen.
\newblock Maximum rank matrix completion.
\newblock {\em Linear Algebra and its Applications}, 288:211--217, 1999.

\bibitem[GI05]{gi05}
J.~F. Geelen and S.~Iwata.
\newblock Matroid matching via mixed skew-symmetric matrices.
\newblock {\em Combinatorica}, 25(2):187--215, 2005.

\bibitem[GIM03]{gim03}
J.~F. Geelen, S.~Iwata, and K.~Murota.
\newblock The linear delta-matroid parity problem.
\newblock {\em Journal of Combinatorial Theory, Series B}, 88(2):377--398,
  2003.

\bibitem[Gur03]{gu03}
L.~Gurvits.
\newblock Classical deterministic complexity of edmonds' problem and quantum
  entanglement.
\newblock In {\em STOC '03: Proceedings of the thirty-fifth annual ACM
  symposium on Theory of computing}, pages 10--19, 2003.

\bibitem[Gur04]{gu04}
L.~Gurvits.
\newblock Classical complexity and quantum entanglement.
\newblock {\em Journal of Computer and System Sciences}, 69(3):448--484, 2004.

\bibitem[HKM05]{hkm05}
N.~J.~A. Harvey, D.~R. Karger, and K.~Murota.
\newblock Deterministic network coding by matrix completion.
\newblock In {\em SODA '05: Proceedings of the sixteenth annual ACM-SIAM
  Symposium on Discrete Algorithms}, pages 489--498, 2005.

\bibitem[HKY06]{hky06}
N.~J.~A. Harvey, D.~R. Karger, and S.~Yekhanin.
\newblock The complexity of matrix completion.
\newblock In {\em SODA '06: Proceedings of the seventeenth annual ACM-SIAM
  Symposium on Discrete Algorithms}, pages 1103--1111, 2006.

\bibitem[IM83]{im83}
O.~Ibarra and S.~Moran.
\newblock Probabilistic algorithms for deciding equivalence of straight-line
  programs.
\newblock {\em JACM}, 30(1):217--228, 1983.

\bibitem[KI03]{ki03}
V.~Kabanets and R.~Impagliazzo.
\newblock Derandomizing polynomial identity tests means proving circuit lower
  bounds.
\newblock In {\em STOC '03: Proceedings of the thirty-fifth annual ACM
  Symposium on Theory of Computing}, pages 355--364, 2003.

\bibitem[Lov79]{l79}
L.~Lov\'{a}sz.
\newblock On determinants, matchings and random algorithms.
\newblock In {\em FCT '79: Fundamentals of Computation Theory}, pages 565--574,
  1979.

\bibitem[Lov89]{l89}
L.~Lov\'{a}sz.
\newblock Singular spaces of matrices and their applications in combinatorics.
\newblock {\em Bol. Soc. Braz. Mat}, 20:87--99, 1989.

\bibitem[Mur00]{m00}
K.~Murota.
\newblock {\em Matrices and Matroids for Systems Analysis}.
\newblock Springer-Verlag, 2000.

\bibitem[Pie82]{p82}
R.~S. Pierce.
\newblock {\em Associative Algebras}.
\newblock Springer-Verlag, New York, 1982.

\bibitem[Sch80]{s80}
J.~T. Schwartz.
\newblock Fast probabilistic algorithms for verification of polynomial
  identities.
\newblock {\em JACM}, 27(4):701--717, 1980.

\bibitem[Val79]{v79}
L.~G. Valiant.
\newblock Completeness classes in algebra.
\newblock In {\em STOC '79: Proceedings of the eleventh annual ACM Symposium on
  Theory of Computing}, pages 249--261, 1979.

\bibitem[Zip79]{z79}
R.~Zippel.
\newblock Probabilistic algorithms for sparse polynomials.
\newblock {\em Symbolic and Algebraic Computation}, pages 216--226, 1979.

\end{thebibliography}

%\newpage
%\appendix
\end{document}